\documentclass[onecolumn,12pt]{IEEEtran}

\usepackage{graphicx}  

\usepackage{subfigure} 

\usepackage{stfloats}  

\usepackage{amsmath,amsfonts,amssymb,amscd}  
\usepackage[ruled,vlined]{algorithm2e} 

\usepackage{graphicx}
\usepackage{color}

\newtheorem{theorem}{\bfseries Theorem}
\newtheorem{lemma}{\bfseries Lemma}

\newtheorem{definition}{\bfseries Definition}
\newtheorem{prop}{\bfseries Proposition}

\usepackage{pdfsync}
\usepackage{epstopdf}

\newcommand{\comment}[1]{}
\usepackage{setspace}

\newcommand{\scriptt}{\mathcal{T}}
\newcommand{\scripti}{\mathcal{I}}
\newcommand{\scripto}{\mathcal{O}}
\newcommand{\scriptv}{\mathcal{V}}
\newcommand{\scripte}{\mathcal{E}}
\newcommand{\scriptm}{\mathcal{M}}

\newcommand{\scriptw}{\mathcal{W}}
\newcommand{\newy}{\tilde{y}}
\newcommand{\newz}{\tilde{z}}

\title{\LARGE{Distributed Algorithms for Consensus and Coordination in the Presence of Packet-Dropping Communication Links\\
\Large{Part II: Coefficients of Ergodicity Analysis Approach }}}
 \author{\authorblockN{Nitin H. Vaidya,~\IEEEmembership{Fellow,~IEEE}} \\
\authorblockN{Christoforos N. Hadjicostis,~\IEEEmembership{Senior Member,~IEEE}}\\
\authorblockN{Alejandro~D.~Dom\'{i}nguez-Garc\'{i}a,~\IEEEmembership{Member,~IEEE}}\\
\authorblockN{~} \\
\authorblockN{September 28, 2011}
\thanks{University of Illinois at Urbana-Champaign.  Coordinated Sciences Laboratory technical report UILU-ENG-11-2208 (CRHC-11-06)}
\thanks{N. H. Vaidya and A. D. Dom\'{i}nguez-Garc\'{i}a  are with the Department of Electrical and Computer Engineering at the University of Illinois at Urbana-Champaign, Urbana, IL 61801, USA. E-mail:  \{nhv, aledan\}@ILLINOIS.EDU.}
\thanks{C. N. Hadjicostis is with the Department of Electrical and Computer Engineering at the University of Cyprus, Nicosia, Cyprus, and also with the Department of Electrical and Computer Engineering at the University of Illinois at Urbana-Champaign, Urbana, IL 61801, USA. E-mail:  chadjic@UCY.AC.CY.}
\thanks{The work of A. D. Dom\'{i}nguez-Garc\'{i}a was supported in part by  NSF under Career Award ECCS-CAR-0954420. The work of C. N. Hadjicostis was supported in part by the  
European Commission (EC) 7th Framework Programme (FP7/2007-2013) 
under grant agreements INFSO-ICT-223844 and PIRG02-GA-2007-224877. The work of N. H. Vaidya was supported in part by Army Research Office grant W-911-NF-0710287 and NSF Award 1059540. Any opinions, findings, and conclusions or recommendations expressed here are those of the authors and do not necessarily reflect the views of the funding agencies or the U.S. government.}

}

\begin{document}

\markboth{Coordinated Sciences Laboratory technical report UILU-ENG-11-2208 (CRHC-11-06)}{}

\maketitle

\begin{abstract}
In this two-part paper, we consider multicomponent systems in which each component  can iteratively exchange information with other components in its neighborhood in order to compute, in a distributed fashion,  the average of the components' initial values or some other quantity of interest (i.e., some function of these initial values).     In particular, we study an  iterative algorithm for computing the average of the initial values of the nodes. In this algorithm,  each component  maintains two sets of variables that are updated via two identical linear iterations. The average of the initial values of the nodes can be asymptotically computed by each node as the ratio of two of the variables it maintains. In the first part of this paper, we show how the update rules for the two sets of variables can be enhanced so that  the algorithm becomes tolerant to communication links that may drop packets, independently among them and independently between different transmission times. In this second part,  by rewriting  the collective dynamics of both iterations, we show  that the resulting system is mathematically equivalent to a finite inhomogenous Markov chain whose transition matrix takes one of finitely many values at each step. Then, by using e a coefficients of ergodicity approach, a method commonly used for convergence analysis of Markov chains,  we prove convergence of the robustified consensus scheme. The analysis suggests that similar convergence should hold under more general conditions as well.
\end{abstract}
\onehalfspace
~\\
~\\
~\\
~\\
\begin{small}\textbf{Note to readers:} Section~\ref{s_intro} discusses the relation between Part II (this report) and the companion Part I of the report, and discusses some related work. The readers may skip Section~\ref{s_intro} without a loss of continuity.\end{small}

\newpage

\section{Introduction}
\label{s_intro}

The focus of  this paper is to analyze the convergence of the robustified double-iteration\footnote{In this second part we will also refer to this algorithm as ``ratio consensus" algorithm and will   use  both denominations interchangeably.} algorithm for average consensus introduced in Part I, utilizing a different framework that allows us to move away from the probabilistic model describing the availability of communication links of Part I. More specifically, instead of focusing on the dynamics of the first and second moments of the two iterations to establish convergence as done in Part I,  we consider a framework that builds upon the theory of finite inhomogenous Markov chains. In this regard, by augmenting the communication graph, we will show that the collective dynamics of each of the two iterations can be rewritten in such a way that the resulting system is mathematically equivalent to a finite inhomogenous Markov chain whose transition matrix takes values from a finite set of possible matrices. Once the problem is recasted in this fashion, tools, such as coefficients of ergodicity, commonly used in the analysis of inhomogenous Markov chains (see, e.g., \cite{Se:06}) are used to prove the convergence of the algorithm. 

Recalling from Part I, when the communication network is perfectly reliable (i.e., in the absence of packet drops),  the collective dynamics of the linear iterations can be described by a discrete-time transition system with no inputs in which the transition matrix is column stochastic and primitive. Then, each node  runs two identical copies of a linear iteration, with each iteration initialized differently depending on the problem to be solved. This double-iteration algorithm is a particular instance of the algorithm in \cite{Benezit:10} (which is a generalization of the algorithm proposed in  \cite{KeDo:03}), where the matrices describing each linear iteration are allowed to vary as time evolves, whereas in our setup (for the ideal case when there are no communication link failures) the  transition matrix is fixed over time. In general, the algorithm described above is   not robust against packet-dropping communication links. It might be possible to robustify it by introducing message delivery acknowledgment mechanisms and retransmission mechanisms, but this has certain overhead and drawbacks as discussed in Section~\ref{considerations}. Also, in a pure broadcast system, which is the communication model we assume in this work, it is easy to see that the double-iteration algorithm above will not work properly. The mechanism we proposed in Part I to robustify the double iteration algorithm was for each node $i$ to keep track of three quantities of interest: i) its own internal state (as captured by the state variables maintained in the original double iteration scheme of \cite{Benezit:10,DoHa:10}; ii) an auxiliary variable that accounts for  the total mass broadcasted so far by node $i$ to (all of) its neighbors; and  iii) another auxiliary variable that accounts for the total received mass from each node $j$  that sends information to  node $i$. The details of the algorithm are provided in Section~\ref{s_robust}, but the key in analyzing convergence of the algorithm is to  show that the collective system dynamics can be rewritten by introducing additional nodes---virtual buffers---that account for the difference between these two auxiliary variables. The resulting enhanced system is equivalent to  an inhomogenous Markov chain whose transition matrix takes values from a finite set. 

As discussed in Part I, even if relying on the ratio of two linear iterations, our work is different from the work in \cite{Benezit:10} in terms of both the communication model and also the nature of the protocol itself. In this regard, a key premise in \cite{Benezit:10} is that stochasticity of the transition matrix must be maintained over time, which requires sending nodes to know the number of nodes that are listening, suggesting that i) either the communication links are perfectly reliable, or ii) there is some acknowledgment and retransmission mechanism that ensures messages are delivered to the listening nodes at every round of information exchange. In our work, we remove both assumptions, and assume a pure broadcast model without acknowledgements and retransmissions. It is very easy to see that in the presence of lossy communication links, the algorithm in \cite{Benezit:10} does not solve the average consensus problems as stochasticity of the transition matrix is not preserved over time. Thus, as mentioned above, the key in the approach we follow  to analyze convergence is to augment the communication graph by introducing additional nodes, and to establish the correctness of the algorithms and establish that the collective dynamics of the resulting system is equivalent to  a finite inhomogenous Markov chain with  transition matrix that values values from a finite set. Once the system is rewritten in this fashion, the robust algorithm for ratio consensus reduces to  a similar setting to the one in  \cite{Benezit:10}, except for the fact that some of the the resulting transition matrices might not have positive diagonals, which is required for the proof in \cite{Benezit:10}. Thus, in this regard, our approach may  be also viewed as a generalization of the main result    in \cite{Benezit:10}.

The idea of augmenting the communication graph has been used in consensus problems to study the impact of bounded (fixed and random) communication delays  \cite{Cao.Morse.ea2008,Nedic:2010,TsRa:11}. In our work, the augmented communication graph that results from rewriting the collective system dynamics has some similarities to the augmented communication graph in \cite{TsRa:11}, where  the   link from node $i$ to node $j$  is replaced by several paths  from  node  $i$ to node  $j$, in order to mimic the effect of communication delays. In particular, in \cite{TsRa:11}, for a maximum delay of $B$ steps, $B$ paths are added in parallel with the single-edge path that captures the non-delayed message transmission. The added path corresponding to delay $b$ ($1\leq b\leq B$) has $b$ nodes, for a total of $B(B+1)/2$ additional nodes capturing the effect of message  transmission delays from node  $i$ to node  $j$. At every time step, a message from  node $i$ to node $j$ is randomly routed through one of these paths; the authors assume  for simplicity that each of the paths is activated with probability $1/B$. For large communication graphs, one of the drawbacks of this model is the explosion in the number of nodes to be added to the communication graph to model the effect of delays. In our work, for analysis purposes, we also use the idea of augmenting the communication graph, but in our case, a single parallel path is sufficient to capture the effect of packet-dropping communication links. As briefly discussed later, it is easy to see that our modeling formalism can also be used to capture random delays, with the advantage over the formalism in  \cite{TsRa:11} that in our model, it is only necessary to add a single parallel path with $B$ nodes (instead of the $B(B+1)/2$ nodes added above) per link in the original communication path, which reduces the number of states added. Additionally, our modeling framework can handle any delay distribution, as long as   the equivalent augmented
network satisfies properties (M1)-(M5) discussed in Section~\ref{ss_matrix}.

In order to make Part II self-contained, we  review several ideas already introduced in Part I, including the double-iteration algorithm formulation over perfectly reliable networks and its robustified version. In Part II, we will embrace the common convention utilized in Markov chains of pre-multiplying the transition matrix of the Markov chain by the corresponding probability vector. 

The remainder of this paper is organized as follows. Section~\ref{s_model} introduces the communication model,  briefly describes the non-robust version of the double-iteration algorithm, and discusses some  issues  that arise  when implementing the double-iteration algorithm in networks with unreliable links. Section~\ref{s_robust} describes the strategy to robustify the double-iteration algorithm against communication link failures. Section~\ref{s_robust_algo} 
reformulates each of the two iterations in the robust algorithm as an inhomogeneous Markov chain. We employ coefficients of ergodicity analysis to characterize the algorithm behavior in Section~\ref{s_ergodicity}.  Convergence  of the robustified double-iteration algorithm is established in Section~\ref{convergence_analysis}. Concluding remarks and discussions on future work are presented in Section~\ref{concluding_remarks}.

\section{Preliminaries}
\label{s_model}
This section describes the communication model we adopt throughout the work, introduces notation,  reviews the double-iteration algorithm that can be used to solve consensus problems when the communication network is perfectly reliable, and discusses issues  that arise  when implementing the double-iteration algorithm in networks with packet-dropping links.

\subsection{Network Communication Model}
The system under consideration consists of a network of $m$ nodes, $\scriptv=\{1,2,\dots, m\}$, each of which has some initial value $v_i,~i=1,2,\dots,m$,
(e.g., a temperature reading). The nodes need to
reach consensus to the average of these initial values
in an iterative fashion. In other words,  the goal is for each node to obtain the value $\frac{\sum_{j=1}^m v_j}{m}$ in a distributed fashion. 
We assume a synchronous\footnote{We later discuss how the techniques we develop for reaching consensus using the double iteration algorithm in the presence of packet-dropping links naturally lead to an asynchronous computation setup.} system in which time is divided
into {\em time steps} of fixed duration.
The nodes in the network are connected by
a certain directed network. More specifically, a directed link $(j,i)$ is said to ``exist'' if 
transmissions from node $j$ can be received by node $i$ infinitely often
over an infinite interval.
Let $\scripte$ denote the set of all directed links that exist in the network.
For notational convenience, we take that $(i,i)\in\scripte,~\forall i$, so that a self-loop exists
at each node.
Then, graph $\mathcal{G}=(\scriptv,\scripte)$ represents the network connectivity.
Let us define $\scripti_i=\{ j ~|~ (j,i) \in \scripte\}$
and $\scripto_i=\{j~|~(i,j)\in\scripte\}$. Thus,
$\scripti_i$ consists of all nodes from whom node $i$ has
incoming links, and $\scripto_i$ consists of all nodes to whom node
$i$ has outgoing links. 
For a set $S$, we will denote the cardinality of set $S$ by $|S|$.
The outdegree of node $i$, denoted as $D_i$, 
is the size of set $\scripto_i$, thus,
$D_i=|\scripto_i|$. 
Due to the assumption that all nodes have self-loops,
$i\in \scripti_i$ and $i\in \scripto_i$, $\forall i\in\scriptv$.
We assume that graph $\mathcal{G}=(\scriptv,\scripte)$ is strongly
connected. Thus, in $\mathcal{G}=(\scriptv,\scripte)$, there exists a directed path from any
node $i$ to any node $j$, $\forall i,j\in\scriptv$ (although it is possible
that the links on such a path between a pair of nodes may not
all be simultaneously reliable in a given time slot).

The iterative consensus algorithms considered
in this paper assume that, at each step of the iteration, each node transmits
some information to all the nodes to whom it has a reliable directed link
during that iteration (or ``time step''). 
The iterative consensus algorithm summarized in Section~\ref{s_ratio} assumes
the special case wherein all the links
are {\em always reliable} (that is, all links are reliable in every time step).
In Section~\ref{s_robust}, and beyond, we consider
a network with potentially unreliable links. Our work on iterative consensus over unreliable links is motivated
by the presence of such links in wireless networks.
Suppose that the nodes in our network 
communicate over wireless links, with the node locations being
fixed. In such a wireless network, each node should 
generally be able to communicate with the other nodes in its
vicinity. However, such transmissions may not always be reliable,
due to channel fading and interference from other sources.
To make our subsequent discussion  precise, we will assume that
a link $(i,j)$ exists (i.e., $(i,j)\in\scripte$)
only if each transmission from $i$
is successfully received by node $j$ with probability $q_{ij}$
($0<q_{ij}\leq 1$).
We assume that successes of transmissions on different
links are independent of each other; also,
successes of different transmissions
on any given link are independent of each other. As we will see, these independence assumptions can be partially relaxed but we adopt them at this point for simplicity.

We assume that all transmissions from any node $i$
are {\em broadcasts},\footnote{As
elaborated later, the results in this paper can also be applied
in networks wherein the transmissions are unicast (not broadcast).} in the sense that, every node
$j$, such that $(i,j)\in\scripte$, may receive $i$'s
transmission with probability $q_{ij}$ independently between nodes and transmission steps.
As seen later, this broadcast property can potentially
be exploited to make communication more efficient, particularly when a
given node $i$ wants to send identical information to all the nodes
in $\scripto_i$. 
When node $i$ broadcasts a message to its neighbors,
the reliabilities of receptions at different nodes in $\scripto_i$
are mutually independent. Each node $i$ is assumed to be aware of the value of $D_i$
(i.e., the number of nodes in $\scripto_i$), and 
the identity of each node in set $\scripti_i$.
This information can be learned using {\em neighbor discovery}\,
mechanisms used in wireless ad hoc or mesh
networks. Note that node $i$ does not necessarily know whether transmissions to nodes in $\mathcal{O}_i$ are successful.

\subsection{Ratio Consensus Algorithm in Perfectly Reliable Communication Networks}
\label{s_ratio}

In this section, we summarize a consensus algorithm for a special
case of the above system, wherein all the links in the network
are {\em always reliable} (that is, reliable in every time step).
The ``ratio consensus'' algorithm presented here performs 
two iterative computations in parallel, with the solution of the
consensus algorithm being asymptotically obtained as the {\em ratio}\, of the
outcome of the two parallel iterations.
We will refer to this approach as {\em ratio consensus}.
In prior literature, similar approaches have also been called
{\em weighted consensus} \cite{Benezit:10,KeDo:03}.

Each node $i$ maintains at iteration $k$ state variables $y_k[i]$
and $z_k[i]$. At each time step $k$, each node $i$ updates
its state variable as follows:
\begin{eqnarray}
y_k[i] & = & \sum_{j\in \scripti_i} ~ y_{k-1}[j] \, / \, D_j ~, ~~~~~~ k\geq 1,\\
z_k[i] & = & \sum_{j\in \scripti_i} ~ z_{k-1}[j] \, / \, D_j ~,~~~~~~ k\geq 1,
\end{eqnarray}
where $y_0[j]=v_j,~\forall j=1,\dots,m$, and $z_0[j]=1,~\forall j=1,\dots,m$.

To facilitate implementation of the above iterations,
at time step $k$, each node $i$ broadcasts a message containing values
$y_{k-1}[i]/D_i$ and $z_{k-1}[i]/D_i$ to each node in $\scripto_i$, and awaits reception of
a similar message from each node in $\scripti_i$. When   node $i$
has received, from each node $j\in \scripti_i$, a value (namely,
$y_{k-1}[j]/D_j$ and $z_{k-1}[i]/D_j$) at step $k$, node $i$
performs the above update of its state variables (by simply summing the corresponding values).
Hereafter, we will use the phrase ``message $v$" to mean ``message
containing value $v$".

The above two iterations are represented in
a matrix notation in (\ref{e_y_fixed})
and (\ref{e_z_fixed}),
where $y_k$ and $z_k$ are row vectors of size $m$,
and $M$ is an $m\times m$ primitive matrix\footnote{A finite square matrix $A$ is said to be 
{\em primitive}\, if for some positive integer
$p$, $A^p>0$, that is, $A^p[i,j] > 0,~\forall i,j$.},
such that $M[i,j]=1/D_i$ if $j\in \scripto_i$
and 0 otherwise. Compactly, we show
\begin{eqnarray}
y_{k} & = & y_{k-1} ~ M  ,~~~~~~~~~~ k \geq 1~, \label{e_y_fixed} \\
z_{k} & = & z_{k-1} ~ M , ~~~~~~~~~~ k \geq 1. \label{e_z_fixed}
\end{eqnarray}
It is assumed that $z_0[j] = 1$ and $y_0[j] =v_j$ are the initial values 
at each node $j\in\scriptv$.
Each node $i$ calculates, at each time step $k$, the ratio
\[
v_k[i] = \frac{y_k[i]}{z_k[i]} \; .
\]
For the transition matrix $M$, (a) $M[i,j]\geq 0$,
and (b) for all $i$, $\sum_j M[i,j]=1$. Any matrix that satisfies
these two conditions is said to be a {\em row stochastic}\, matrix.
It has been shown in \cite{DoHa:10} that $v_k[i]$ asymptotically converges
to the average of the elements of $y_0$, provided that $M$ is {\em primitive}
and {\em row stochastic}. That is, if $M$ is a primitive row stochastic
matrix, then
\begin{eqnarray}
\lim_{k\rightarrow \infty} ~ v_k[i] = 
\frac{ \sum_j ~ y_0[j] }{m}, \;~\forall i\in\scriptv,
\label{e_fixed}
\end{eqnarray}
where $m$ is the number of elements in vector $y_0$.

%

\comment{++++++++++
NOTE: IS THE ABOVE SHOWN SPECIFICALLY FOR THE CASE WHEN MATRIX ENTRIES ARE
$1/D_i$, OR FOR ANY STOCHASTIC MATRX?

NOTE: WE HAVE NOT REALLY SHOWN THE ABOVE. WE SHOWED IT FOR $Z_0=ONES$ IN WHICH CASE THE RATIO IS ALWAYS WELL DEFINED BECAUSE IN OUR CASE WE ALWAYS HAVE $Z[I]>0$.

NEED TO MAKE THE TEXT ABOVE CONSISTENT WHATEVER WAS SHOWN IN THE PREVIOUS
PAPER.
+++++++++++++++}

\subsection{Implementation Aspects of Ratio Consensus Algorithm in the Presence of  Unreliable Links} \label{considerations}

Let us consider how we might implement
iterations (\ref{e_y_fixed}) and (\ref{e_z_fixed}) in a wireless network.
Since the treatment for the $y_k$ and $z_k$ iterations is
similar, let us focus on the $y_k$ iteration for now.
Implementing (\ref{e_y_fixed}) requires that,
at iteration $k$ (to compute $y_k$), node $i$ should
transmit message $y_{k-1}[i]\, M[i,j]$ to
each node $j\in \scripto_i$.
Conveniently, for all $j\in \scripto_i$, the values $M[i,j]$ are
identical, and equal to $1/D_i$. Thus, node $i$ needs to send
message $y_{k-1}[i]/ D_i$ to each
node in $\scripto_i$. Let us define
\[ \mu_k[i] \equiv y_{k-1}[i] \, / D_i~, ~~~~~ k\geq 1.
\]

In a wireless network, the two approaches described next may be used by node
$i$ to transmit message $\mu_k[i]$ to all the nodes in $\scripto_i$.

\subsubsection*{Approach 1} In this approach, each node $i$ ensures that its
message $\mu_k[i]$ is delivered reliably to all the nodes in $\scripto_i$.
One way to achieve this goal is as follows. Node $i$ can 
broadcast the message $\mu_k[i]$ on the wireless channel, and then wait
for acknowlegdements (ack) from all the nodes in $\scripto_i$.
If such acks are not
received from all nodes in $\scripto_i$ within some timeout interval,
then $i$ can retransmit the message.
This procedure will be repeated until acks are received from all the intended
recipients of $\mu_k[i]$.
This procedure ensures that the message is received by each node in $\scripto_i$
reliably in each step $k$ of the iteration. However, as an undesirable side-effect,
the   time required to guarantee the reliable delivery to all the neighboring
nodes is not fixed. In fact, this  time can
be arbitrarily large with a non-zero probability, if
each transmission on a link $(i,j)\in\scripte$ is reliable with probability
$q_{ij}<1$. Different nodes may require different amounts of time to
reliably deliver their message to their intended recipients.
Thus, if a fixed finite interval of time is allocated for each step $k$,
then it becomes difficult to guarantee that the iterations will be
always performed {\em correctly} (because some messages 
may not be delivered within the fixed time interval).

\subsubsection*{Approach 2} Alternatively, each node $i$ may just broadcast its message
$\mu_k[i]$ once in time step $k$, and hope that
all the nodes in $\scripto_i$ receive it reliably.
This approach has the advantage that each step of the iteration
can be performed in a short (and predictable) time interval.
However, it also has the undesirable property that all the nodes
in $\scripto_i$ may not receive the message (due to link unreliability),
and such nodes will not be able to update their state correctly.
It is important to note that, since there are no acknowlegements
being sent,
a node $i$ cannot immediately know whether a node $j\in\scripto_i$ has
received $i$'s message or not.

Considering the shortcomings of the above two approaches, it appears
that an alternative solution is required.
Our solution to the problem (to be introduced in Section~\ref{s_robust}) is to maintain  {\em additional}\,
state at each node, and utilize this state to mitigate the
detrimental impact of link unreliability. To put it differently,
the additional state can be used to design an iterative
consensus algorithm {\em robust}\, to link unreliability.
In particular, the amount of state maintained by each
node $i$ is proportional to $|\scripti_i|$. In a large scale wireless network (i.e,
with large $m$) with nodes spread over large space, we would expect
that for any node $i$, $|\scripti_i|<<m$. In such
cases, the small increase in the amount of state
is a justifiable cost to achieve robustness in presence
of link unreliability.

Although $M[i,j]$ is identical (and equal to $1/D_i$)
for all $j\in O_i$ in our example above, this is not necessary.
So long as $M$ is a primitive row stochastic matrix, the
above iteration will converge
to the correct consensus value (provided that the transmissions
are always reliable).
Thus, it is possible that in a given iteration, node $i$ may want to
send different messages to different nodes in $O_i$. This goal can be
achieved by performing unicast operation to each node in $O_i$.
In this situation as well, two approaches analogous to Approaches 1 and
2 may be used. The first approach would be to reliably deliver the unicast
messages, using as many retransmissions as necessarys. The
second approach may be to transmit each message just once.
In both cases, it is possible that the iterations may not be
performed correctly. To simplify the discussion in this paper, we assume
that each node $i$ needs to transmit identical message to
the nodes in $\scripto_i$. However, it is easy to  extend the proposed
scheme so that it is applicable to the more general scenario as well.

\section{Robustification of Ratio Consensus Algorithm}
\label{s_robust}

In this section, we present the proposed ratio consensus algorithm
that is robust in presence of link unreliability. The correctness
of the proposed algorithm is established in Section~\ref{convergence_analysis}.
As before, each node maintains state variables $y_k[i]$ and $z_k[i]$.
Additional state maintained at each node will be defined soon.
Iterative computation is performed to maintain $y_k$ and $z_k$.
For brevity,   we will focus on presenting the iterations for $y_k$, but iterations for $z_k$
are analogous, with the difference being in the initial state. The initial values of $y$ and $z$ are assumed\footnote{The assumption that $y_0[i]\geq 0,~\forall i$,
can be relaxed, allowing for arbitrary values for $y_0[i]$.} to
satisfy the following conditions:
\begin{enumerate}
\item $y_0[i]\geq 0, ~ \forall i$,
\item $z_0[i]\geq 0, ~ \forall i$,  
\item $ \sum_i z_0[i] > 0$.
\end{enumerate}
Our goal for the robust iterative consensus algorithm is to allow
each node $i$ to compute (asymptotically) the ratio
\[
\frac{ \sum_i y_0[i]}{\sum_i z_0[i]}.
\]
With a suitable choice of $y_0[i]$ and $z_0[i]$, different functions
may be calculated \cite{DoHa:10}. In particular, if the initial input
of node $i$ is denoted as $v_i$, then by setting $y_0[i]=w_iv_i$
and $z_0[i]=w_i$, where $w_i\geq 0,~\forall i$,
the nodes can compute the weighted average $\frac{\sum_i w_iv_i }{\sum_i w_i}$; 
with $w_i=1,~\forall i\in\scriptv,$ the nodes calculate average consensus.

\subsection{Intuition Behind the Robust Algorithm}
\label{ss_intuition}

To aid our presentation, let us introduce the notion of ``mass.'' The
initial value $y_0[i]$ at node $i$ is to be viewed as its initial mass.
If node $i$ sends a message $v$ to another node $j$, that can be
viewed as a ``transfer'' of an amount of mass equal to $v$ to node $j$.
With this viewpoint, it helps to think of each step $k$
as being performed over a non-zero interval of time.
Then, $y_k[i]$ should be viewed
as the mass at node $i$ at the {\em end}\, of time step $k$
(which is the same as the  {\em start} of step $k+1$).
Thus, during step $k$, each node $i$ transfers (perhaps unsuccessfully,
due to unreliable links) some mass
to nodes in $\scripto_i$, the amount being a function of $y_{k-1}[i]$.
The mass $y_k[i]$ is the accumulation of the mass that $i$ receives
in messages from nodes in $\scripti_i$ during step $k$.

Now, $\sum_i y_0[i]$ is the total
mass in the system initially. If we implement iteration
(\ref{e_y_fixed}) in the absence of packet drops, then for all iterations $k$
\[
\sum_i y_k[i] = \sum_i y_0[i].
\]
That is, the total mass in the system remains constant.
This invariant is maintained because $M$ is a row stochastic matrix.
However, if a message $v$ sent by node $i$ is not received by some
node $j\in\scripto_i$, then the mass in that message is ``lost,''
resulting in  reduction of the total mass in the system. 

Our robust algorithm is motivated by the desire to avoid the loss
of mass in the system, even in the presence of unreliable links.
The proposed algorithm uses Approach 2 for transmission of messages.
In particular, in our algorithm (and as in the original ratio
consensus), at each step $k$, each node $i$
wants to transfer $\mu_k[i]=y_{k-1}[i]/D_i$ amount of mass to each
node in $\scripto_i$. For this purpose, node $i$ broadcasts\footnote{In the
more general case, node $i$ may want to transfer different amounts
of mass to different nodes in $\scripto_i$. In this case, node $i$
may send (unreliable) unicast messages to these neighbors. The treatment in
this case will be quite similar to the restricted case assumed in our
discussion, except that node $i$ will need to separately track mass transfers to each of its out-neighbors.}
message $\mu_k[i]$.
To make the algorithm robust, let us assume that,
for each link $(i,j)\in\scripte$, a ``virtual buffer''
is available to store the mass that is ``undelivered'' on the link.
For each node $j\in\scripto_i$, there are two possibilities:
\begin{itemize}
\item [(P1)] Link $(i,j)$ is not reliable in slot $k$: In this case,
	message $\mu_k[i]$ is {\em not}\, received by node $j$.
	Node $i$ believes that it has transferred the mass to $j$
	(and thus, $i$ does not include that mass in its own state
	$y_k[i]$), and at the same time, that mass is
	not received at node $j$, and therefore, not
	included in $y_k[j]$.
 	Therefore, let us view this missing mass as being ``buffered
	on'' link $(i,j)$ in a {\em virtual buffer}. The virtual buffer for each directed link $(i,j)$ will be
	viewed as a {\em virtual node} in the network.
	Thus, when link $(i,j)$ is unreliable, the mass is
	transferred from node $i$ to ``node'' $(i,j)$,
	instead of being transferred to node $j$. Note that when link $(i,j)$ is unreliable, node $j$ neither
	receives  mass directly from node $i$, nor from the virtual buffer $(i,j)$.

\item [(P2)] Link $(i,j)$ is reliable in slot $k$:
	In this case, message $\mu_k[i]$ is received by node $j$.
	Thus, $\mu_k[i]$ contributes to $y_k[j]$.
	In addition, all the mass buffered in the virtual buffer $(i,j)$
	will also be received by node $j$, and this mass will also
	contribute to $y_k[j]$.
	We will say that buffer $(i,j)$ ``releases'' its mass to node $j$.
\end{itemize}

We capture the above intuition by building an ``augmented'' network
that contains all the nodes in $\scriptv$, and also contains
additional virtual nodes, each virtual node corresponding to
the virtual buffer for a link in $\scripte$.
Let us denote the augmented networks by $\mathcal{G}^a=(\mathcal{V}^a,\mathcal{E}^a)$ where
$\mathcal{V}^a=\scriptv \cup \scripte$ and 
\[ \mathcal{E}^a=\scripte \cup \{((i,j),j)~|~(i,j)\in\scripte\}
\cup \{i,(i,j)~|~(i,j)\in\scripte\}.\]

In case (P2) above, the mass sent by node $i$,
and the mass released from the virtual buffer $(i,j)$, both contribute
to the new state $y_k[j]$ at node $j$. In particular, it will
suffice for node $j$ to only know the {\em sum}\, of the
mass being sent by node $i$ at step $k$ and the mass being released
(if any) from buffer $(i,j)$ at step $k$.
In reality, of course, there is no virtual buffer to hold the mass
that has not been delivered yet. However, an equivalent mechanism
can be implemented by introducing additional state at each node
in $\scriptv$,
which exploits the above observation. This is what we explain in the next section.

\subsection{Robust Ratio Consensus Algorithm}
\label{ss_robust}

We will mitigate the shortcomings of Approach 2 described in Section~\ref{considerations} by changing our iterations
to be tolerant to missing messages. The modified scheme has the following
features:
\begin{itemize}
\item Instead of transmitting message $\mu_k[i]=y_{k-1}[i]/D_i$ at step $k$,
 each node $i$ broadcasts at step $k$ a message with value $\sum_{j=1}^k \mu_k[i]$,
 denoted as $\sigma_k[i]$. Thus, $\sigma_k[i]$ is the total mass that node $i$
 wants to transfer to each node in $\scripto_i$ through the
 first $k$ steps.
\item Each node $i$ maintains, in addition to state variables $y_k[i]$ and $z_k[i]$,
	also a state variable $\rho_k[j,i]$ for each node $j\in \scripti_i$; 
	$\rho_k[j,i]$ is the total mass that node $i$ has received
	either directly from node $j$, or via virtual buffer $(j,i)$,
	through step $k$.
\end{itemize}
The computation performed at node $i$ at step $k\geq 1$ is as follows.
Note that $\sigma_0[i]=0$, $\forall i\in\scriptv$ and $\rho_0[i,j]=0,~\forall (i,j)\in\scripte$.
\begin{eqnarray}
\sigma_k[i] & = & \sigma_{k-1}[i] + y_{k-1}[i]/D_i, \label{e_sigma}\\
\rho_k[j,i] & = & \left\{ \begin{tabular}{ll}
		$\sigma_k[j]$, & \mbox{~if $(j,i)\in \scripte$ and message $\sigma_k[j]$ is received by $i$
			from $j$ at step $k$,}\\
		$\rho_{k-1}[j,i],$ & \mbox{~if $(j,i)\in \scripte$ and no message is received by $i$
                        from $j$ at step $k$},
		\end{tabular}\right. \label{e_rho} \\
y_k[i] & = & \sum_{j\in \scripti_i} (\rho_k[j,i] - \rho_{k-1}[j,i]). \label{e_y_rand}
\label{e_y_k}
\end{eqnarray}
When link $(j,i)\in\scripte$ is reliable, $\rho_k[j,i]$
becomes equal to $\sigma_k[j]$: this is reasonable, because $i$ receives
any new mass sent by $j$ at step $k$, as well as any mass released by
buffer $(j,i)$ at step $k$. On the other hand, when link $(j,i)$ is unreliable,
then $\rho_k[j,i]$ remains unchanged from the previous iteration,
since no mass is received from $j$ (either directly or via
virtual buffer $(j,i)$).
It follows that, the total new mass received by node $i$ at step $k$, 
either from node $j$
directly or via buffer $(j,i)$, is given by 
$
\rho_k[j,i] - \rho_{k-1}[j,i],
$
which explains (\ref{e_y_k}).\footnote{As per the algorithm specified above,
observe that the values of $\sigma$ and $\rho$  increase monotonically
with time. This can be a concern for a large number of steps in practical implementations. However,
this concern can be mitigated by ``resetting" these values, e.g., via the exchange of additional
information between neighbors
(for instance, by piggybacking cumulative acknowledgements, which
will be delivered whenever the links operate reliably).}

\section{Robust Algorithm Formulation as an Inhomogeneous Markov Chain} \label{InhMC}
\label{s_robust_algo}

In this section, we reformulate each iteration performed by the robust algorithm as an inhomogeneous Markov chain whose transition matrix takes values from a finite set of matrices. We will also discuss  some properties of these matrices, and analyze the behavior of their products, which helps
 in establishing the convergence of the robustified ratio consensus algorithm.

\subsection{Matrix Representation of Each Individual  Iteration}
\label{ss_matrix}

The matrix representation is obtained by observing an equivalence
between the iteration in  (\ref{e_sigma})--(\ref{e_y_rand}), and an iterative algorithm (to be introduced soon)
defined on  the augmented network described in Section~\ref{ss_intuition}. 
The vector state of the augmented network consists of $n=m+|\scripte|$ elements,
corresponding to the mass held by each of the $m$ nodes,
and the mass held by each of the
$|\scripte|$ virtual buffers: these $n$ entities are represented
by as many nodes in the augmented network. 


With a slight abuse of notation, let us denote by $y_k$
the state of the nodes in the augmented network $\mathcal{G}^a$.
The vector $y_k$ for $\mathcal{G}^a$ is an augmented version of
$y_k$ for $\mathcal{G}$. In addition to $y_k[i]$
for each $i\in\scriptv$, the augmented $y_k$ vector also includes
elements $y_k[(i,j)]$ for each $(i,j)\in\scripte$, with
$y_0[(i,j)]=0$.\footnote{Similarly, $z_0[(i,j)]=0$.} Due to the manner in which the $y_k[i]$'s are updated,
$y_k[i],~i\in\scriptv$, are identical in the original network
and the augmented network; therefore, we do not distinguish between
them. We next translate the iterative algorithm in
(\ref{e_sigma})--(\ref{e_y_k})
into the   matrix form 
\begin{align}
y_k = y_{k-1} M_k, \label{iter_iter}
\end{align}
for appropriately  row-stochastic matrices $M_k$ (to be defined soon) that might vary as the algorithm progresses (but nevertheless take values from a finite set of possible matrices).

Let us define an indicator variable $X_k[j,i]$ for each link $(j,i)\in\scripte$ at each time step $k$ as follows:
\begin{align}
X_k[j,i] = \left\{
	\begin{tabular}{ll}
	1, ~~~~ if link $(j,i)$ is reliable at time step $k$, \\
	0, ~~~~ otherwise.  \label{indicator_var}
	\end{tabular}
	\right.
\end{align}
We will now reformulate the iteration (\ref{e_sigma})--(\ref{e_y_k}) and show how, in fact, it can be described in matrix form as shown in \eqref{iter_iter}, where the matrix transition matrix $M_k$ is a function of the indicator variables defined in  \eqref{indicator_var}. First, by using the indicator variables at time step $k$, as defined in \eqref{indicator_var}, it follows from  \eqref{e_sigma} that
\begin{align}
& \rho_k[j,i]=X_k[j,i]\sigma_k[j]+(1-X_k[j,i])\rho_{k-1}[j,i].  \label{e_rho_reform}
\end{align}
Now, for $k\geq 0$, define $\nu_k[j,i]=\sigma_k[j]-\rho_k[j,i]$
(thus $\nu_0[j,i]=0$). Then, it follows from \eqref{e_sigma} and \eqref{e_rho_reform} that 
\begin{align}
\nu_k[j,i]=(1-X_k[j,i]) \left(\frac{y_{k-1}[j]}{D_j}+\nu_{k-1}[j,i]\right), ~~~~k\geq 1. \label{nu_def}  \end{align}
Also, from \eqref{e_sigma} and \eqref{e_rho_reform}, it follows that \eqref{e_y_k} can be rewritten as 
\begin{align}
& y_k[i] =   \sum_{j\in \scripti_i} X_k[j,i]\left(\frac{y_{k-1}[j]}{D_j}+\nu_{k-1}[j,i]\right), ~~~~k\geq 1. \label{y_k_reform}
\end{align}
At every instant $k$ that the link $(j,i)$ is not reliable, it is easy to see that the variable $\nu_k[j,i]$ increases by an amount equal to the amount that node $j$ wished to send to node $i$, but $i$ never received due to the  link failure. Similarly, at every instant $k$ that the link $(j,i)$ is reliable, the variable $\nu_k[j,i]$ becomes zero and its value at $k-1$ is received by node $i$ as can be seen in \eqref{y_k_reform}.  Thus, from \eqref{nu_def} and \eqref{y_k_reform}, we can think of the variable $\nu_k[j,i]$ as the state of a virtual node that buffers  the mass that node $i$ does not receive from node $j$ every time the link $(j,i)$ fails. It is important to note that the $\nu_k[j,i]$'s are virtual variables  (no node in $\mathcal{V}$ computes $\nu_k$) that just result from combining, as explained above, variables that the nodes in $\mathcal{V}$ compute. The reason for doing this is that the resulting model is equivalent to an inhomogeneous Markov chain. This can be   easily seen  by stacking up \eqref{y_k_reform} for all nodes indexed in $\mathcal{V}$, i.e., the computing nodes, and  \eqref{nu_def} for all virtual buffers $(j,i)$, with $(j,i) \in \scripte$, and rewriting the resulting expressions in matrix form, from where the expression in \eqref{iter_iter} results. 

\subsection{Structure and Properties of the Matrices $M_k$} 
Next, we  discuss the sparsity structure of the $M_k$'s and obtain their entries by inspection of \eqref{nu_def} and \eqref{y_k_reform}. Additionally, we will  explore some properties of the $M_k$'s that will be helpful in the analysis conducted in Section~\ref{ergodicity_analysis} for characterizing  the behavior of each of the individual iterations. 

\subsubsection{Structure of $M_k$}Let us first define the entries in row $i$ of matrix $M_k$ that corresponds to $i\in\scriptv$. For $(i,j)\in\scripte$, there are two possibilities:
$X_k[i,j]=0$ or $X_k[i,j]=1$. If $X_k[i,j]=0$, then the mass
$\mu_k[i]=y_k[i]/D_i$ that node $i$ wants to send to node $j$ is added
to the virtual buffer $(i,j)$. Otherwise, no new mass from
node $i$ is added to buffer $(i,j)$. Therefore,
\begin{eqnarray}
& M_k[i,(i,j)] = (1-X_k[i,j])   / D_i.  \label{entry1}
\end{eqnarray}
The above value is zero if link $(i,j)$ is reliable at step $k$, and $1/D_i$
otherwise.
Similarly, it follows that
\begin{eqnarray}
&M_k[i,j] = X_k[i,j] / D_i ,
\end{eqnarray}
which is zero whenever  link $(i,j)$ is unreliable at step $k$, and $1/D_i$
otherwise.
Observe that for each $j\in \scripto_i$,
\begin{eqnarray}
& M_k[i,j] + M_k[i,(i,j)] = 1/D_i,
\end{eqnarray}
with, in fact, one of the two quantities zero and the other equal to $1/D_i$. For $(i,j)\notin\scripte$, it naturally follows that $M_k[i,j]=0$.
Similarly, 
\begin{eqnarray}
& M_k[i,(s,r)]=0, ~~~~\mbox{whenever ~} i\neq s \mbox{~~\text{and}~~} (s,r)\in\scripte.
\end{eqnarray}
Since $|\scripto_i| = D_i$, all the elements in row $i$ of matrix $M_k$ add up to 1.

Now define row $(i,j)$ of matrix $M_k$, which describes how the mass of the virtual buffer $(i,j)$, for $(i,j)\in\scripte$, gets distributed. When link $(i,j)$ works reliably at time step $k$ (i.e., $X_k[i,j]=1$),
	all the mass buffered on link $(i,j)$ is transferred to node $j$; otherwise,
	no mass is trasferred from buffer $(i,j)$ to node $j$ and the buffer
	retains all its previous mass and increases it by a quantity equal to the mass that node $i$ fail to send to node $j$.
	These conditions are captured by defining $M_k$ entries as follows:
	\begin{eqnarray}
	M_k[(i,j),j] & = & X_k[i,j],  \\
	M_k[(i,j),(i,j)] & = & 1-X_k[i,j].
	\end{eqnarray}
	Also, for obvious reasons,
	\begin{eqnarray}
	M_k[(i,j),p] & = & 0,~ ~ \forall p \neq j, ~~ p\in\scriptv, \\
	M_k[(i,j),(s,r)] & = & 0,~~ \forall (i,j)\neq (s,r),~~(s,r)\in\scripte. \label{entryk}
	\end{eqnarray}  
Clearly, all the entries of the row labeled $(i,j)$ add up to 1, which results in  $M_k$ being a row stochastic matrix for all $k\geq 1$.

\subsubsection{Properties of $M_k$}

Let us denote the set of all possible instances (depending
on the values of the indicator variables $X_k[i,j],~(i,j) \in \mathcal{E},~k\geq 1$) of matrix $M_k$ as $\scriptm$. The matrices in the set $\scriptm$ have the following properties:
\begin{itemize}

\item [(M1)] \textit{The set $\scriptm$ is  finite.}  

Each distinct matrix in
$\scriptm$ corresponds to  different instantiations of the  
indicator variables defined in \eqref{indicator_var}, resulting in exactly $2^{|\scripte|}$ distinct matrices in
$\scriptm$.

\item[(M2)] \textit{Each matrix in $\scriptm$ is a finite-dimensional square row stochastic matrix.}  

The number of rows of each matrix $M_k\in \scriptm$, as defined above,
is $n=m+|\scripte|$, which is finite.
Also, from   (\ref{entry1})--(\ref{entryk}), theses matrices  are square row-stochastic matrices.

\item [(M3)]
\textit{Each positive element of any matrix in $\scriptm$ is lower bounded
by a positive constant.} 

Let us denote  this lower bound as $c$. Then,  due to the manner in which matrices in $\scriptm$
are constructed, we can define $c$ to be the positive
constant obtained as
\[ c = \min_{i,j,M~| M\in\scriptm,M[i,j]>0} ~ M[i,j].
\]

\item [(M4)]  The matrix $M_k,~k \geq 0$, may be chosen to be any matrix
	$M\in\scriptm$ with a non-zero probability. 
	The choice of the transition matrix at each time step is independent and identically distributed (i.i.d.) due to the assumption that link failures are independent (between nodes and time steps).
	
	{\em Explanation:}  The probability
	distribution on $\scriptm$ is a function of the probability 
	distribution on the link reliability. In particular, 
	if a certain $M\in \scriptm$ is obtained when the
	links in $\scripte'\subseteq\scripte$
	are reliable, and the remaining links are unreliable,
	then the probability that $M_k=M$ is equal to
\begin{align}
\Pi_{(i,j)\in\scripte'}~q_{ij} ~ \Pi_{(i,j)\in\scripte-\scripte'}~(1-q_{ij}).
\end{align}

\item [(M5)]
\textit{For each $i\in\scriptv$, there exists a finite positive integer
$l_i$ such that it is possible to find $l_i$ matrices in
$\scriptm$ (possibly with repetition) such that their product (in a chosen
order) is a row stochastic matrix with the column that corresponds
to node $i$ containing strictly positive entries.}

This property states that, for each $i\in\scriptv$, there exists a  matrix 
$T_i^*$, obtained as  the product of $l_i$ matrices in $\scriptm$ that  has
 the following properties:
\begin{eqnarray}
T_i^*[j,i] & > & 0, ~~~~~ \forall~j\in\scriptv, \label{prop1} \\
T_i^*[(j_1,j_2),i] & > & 0, ~~~~~ \forall~(j_1,j_2)\in\scripte. \label{prop2}
\end{eqnarray}  
This follows from the fact that the underlying graph $\mathcal{G}^a$ is strongly connected (in fact, it can be easily shown that $l_i \leq m$).
To simplify the presentation below, and due to the self-loops, we can take $l_i$ to be equal
to a constant $l$, for all $i\in\scriptv$. However, it should be easy to
see that the arguments below can be generalized to the case when the $l_i$'s
may be different.

We can also show that under our assumption for link failures,
there exists a single matrix, say $T^*$, which simultaneously satisfies the conditions in \eqref{prop1}--\eqref{prop2}
for all $i\in\scriptv$.
\comment{+++++++++++++++++++++
In \eqref{indicator_var}, since $\Pr \{X_k[i,j]=1 \}>0,~ \forall k\geq 1, \forall ~(i,j) \in \mathcal{E}$, it follows that, with a non-zero probability, there exits a matrix $M^a \in \mathcal{M}$  that corresponds to all the links in $\scripte$ being reliable, and it is easy to see   from (\ref{entry1})--(\ref{entryk}) that $M^a$  has the following structure
\begin{align}
M^a=\begin{bmatrix} M && 0 \\ Q && 0 \end{bmatrix}, \label{the_matrix_decomp_2}
\end{align}
where  $M \in \mathbb{R}^{m  \times m}$ is the same row stochastic matrix describing   the collective dynamics of the two iterations in \eqref{e_y_fixed} and \eqref{e_z_fixed} that defined the ratio consensus algorithm for the case  when the communication links were assumed to be perfectly reliable. As argued before, since $\mathcal{G}$ is strongly connected, by construction $M$ is primitive with a positive diagonal. The matrix $M^a$ corresponds to a homogenous Markov chain with a single recurrent class and a few transient states. It is  then easy to see that by  multiplying $M^a$ $k$ times with itself, we obtain 
\begin{align}
(M^a)^k=\begin{bmatrix} M^k && 0 \\ QM^{k-1} && 0 \end{bmatrix}, \label{the_matrix_decomp_3}
\end{align}
and since $M$ is a primitive matrix, there exists a positive integer $l$ such that all the entires in $M^l$ are strictly positive. Then, it follows that $T^*=(M^a)^l$ has the  properties in \eqref{prop1} and \eqref{prop2}.


An alternative explanation for this property is as follows.
++++++++++++++++++++++++++++++++
}
 When all the links in the network operate reliably,
network $\mathcal{G}(\scriptv,\scripte)$ is strongly connected (by assumption).
Since $\mathcal{G}$ is strongly connected, there is a directed path between every pair of nodes $i$ and
$j$, i.e.,   $i,j\in\scriptv$ .
In the augmented network $\mathcal{G}^a$, for each
$(i,j)\in\scripte$, there is a link from node $i$ to node $(i,j)$,
and a link from node $(i,j)$ to node $j$. Thus, it should be clear that the
augmented network $\mathcal{G}^a$ is strongly connected as well.
Consider a spanning tree rooted at node 1, such that all the nodes
in $V=\scriptv\cup\scripte$ have a directed path towards node 1, and also
a spanning tree in which all the nodes have directed paths {\em from}
node 1. Choose that matrix, say $M^*\in\scriptm$, which corresponds to all the
links on these two spanning trees, as well as self-loops
at all $i\in\scriptv$, being reliable. If the total number
of links that are thus reliable is $e$, it should be obvious that $(M^*)^e$ will
contain only non-zero entries in columns corresponding to $i\in\scriptv$.
Thus, $l$ defined above may be chosen as $e$.
There are several other ways of constructing $T^*$, some of which may result
in a smaller value of $l$.


\end{itemize}

\section{Ergodicity Analysis of Products of Matrices $M_k$} \label{ergodicity_analysis}
\label{s_ergodicity}

We will next analyze the ergodic behavior of  the \textit{forward product} $T_{k}=M_{1}M_{2}\dots M_{k}= \Pi_{j=1}^{k}M_j$, where $M_j \in \mathcal{M},~\forall j= 1, 2,\dots,k$. Informally defined, weak ergodicity of $T_{k}$ obtains if the rows of $T_{k}$ tend to equalize as $k \rightarrow \infty$.
 In this work, we focus on the weak ergodicity notion, and  establish probabilistic statements pertaining the  ergodic behavior of  $T_{k}$. The analysis   builds upon a large body of literature on products of nonnegative matrices (see, e.g., \cite{Se:06} for a comprehensive account). First, we introduce the basic toolkit  adopted from  \cite{Ha:58,Wo:63,Se:06},  and then use it to analyze the ergodicity of $T_{k}$.

\subsection{Some Results Pertaining  Coefficients of Ergodicity}


Informally speaking, a coefficient  of ergodicity of a matrix $A$ characterizes how different two rows of $A$ are.
For a row stochastic matrix $A$, proper\footnote{Any scalar function $\tau(\cdot)$ continuous on the set of $n \times n$ row stochastic matrices, which satisfies $0 \leq \tau(A) \leq 1$, is said to be a proper coefficient of ergodicity if $\tau(A)=0$ if and only if $A=e^Tv$, where $e$ is the all-ones row vector, and $v\geq0$ is such that $ve^T=1$ \cite{Se:06}.
}
 coefficients of  ergodicity   $\delta(A)$ and $\lambda(A)$ are defined as:
\begin{align}
\delta(A) & :=   \max_j ~ \max_{i_1,i_2}~ | A[i_1,j]-A[i_2,j] |, \label{e_delta} \\
\lambda(A) & :=  1 - \min_{i_1,i_2} \sum_j \min(A[i_1,j], A[i_2,j]). \label{e_lambda}
\end{align}
It  is easy to see that  $0\leq \delta(A) \leq 1$ and $0\leq \lambda(A) \leq 1$, and that the rows are identical if and only if $\delta(A)=0$. Additionally, $\lambda(A) = 0$ if and only if $\delta(A) = 0$.

The next result establishes a relation between the coefficient of ergodicity $\delta(\cdot)$ of a product of row stochastic matrices, and the coefficients of ergodicity $\lambda(\cdot)$ of the individual matrices defining the product. This result will be used in the proof of Lemma~\ref{l_t}. It  was established in \cite{Ha:58} and also follows from the more general statement of Theorem 4.8 in \cite{Se:06}. 

\begin{prop}
\label{p_delta_product}
For any $p$ square row stochastic matrices $A_1,A_2,\dots A_{p-1}, A_p$, 
\begin{align}
\delta(A_1A_2\cdots A_{p-1}A_p) ~\leq ~
 \left( \Pi_{i=1}^{p-1} \lambda(A_i) \right) \delta(A_p) ~\leq ~
 \Pi_{i=1}^p ~ \lambda(A_i). \label{hajnal}
\end{align}
\end{prop}

The result in \eqref{hajnal} is particularly useful to infer ergodicity of a product of matrices from the ergodic properties of the individual matrices in the product. For example, if $\lambda(A_i)$ is less than 1 for all $i$, then $\delta(A_1A_2\cdots A_{p-1}A_p)$ will tend to zero as $p \rightarrow \infty$. We will next introduce an important class of matrices for which $\lambda(\cdot)<1$.

\begin{definition}
A matrix $A$ is said to be a {\em scrambling}\, matrix, if $\lambda(A)<1$ \cite{Se:06}.
\end{definition}

In a scrambling matrix $A$, since $\lambda(A)<1$, for each pair of
rows $i_1$ and $i_2$, there exists a column $j$ (which may depend on
$i_1$ and $i_2$) such that $A[i_1,j]>0$ and $A[i_2,j]>0$, and vice-versa.
As a special case, if any one column of a row stochastic matrix $A$
contains only non-zero entries, then $A$ must be scrambling. 

%

\subsection{Ergodicity Analysis of Iterations of the Robust Algorithm}
We next analyze the ergodic properties of the products of matrices that result from each of the iterations comprising our robust algorithm. Let us focus on just one of the iterations, say  $y_k$, as the treatment  of the $z_k$ iteration is identical. As described in Section~\ref{InhMC}, the progress of the $y_k$ iteration   can be recast as an inhomogeneous Markov chain
\begin{align}
& y_k=y_{k-1}M_k,~k \geq 1, \label{MC_iter}
\end{align}
where $M_k \in \mathcal{M},~\forall k$. As already discussed, the sequence of $M_k$'s that will govern the progress of $y_k$ is determined by communication link availability. (\ref{MC_iter}). Defining $T_k=\Pi_{j=1}^k\, M_j$, we obtain:
\begin{align}
\displaystyle y_{k} & =   y_0 M_1M_2\cdots M_k \nonumber \\
\displaystyle  & = y_0  \Pi_{j=1}^{k}  M_j= y_0T_k, ~ k \geq 1. \label{e_M}
\end{align}
By convention, 
%
$\Pi_{i=k}^0 M_i = I$ for any $k\geq 1$ ($I$ denotes the $n \times n$ identity matrix).

Recalling the constant $l$ defined in (M5),
define $W_k$ as follows,
\begin{eqnarray}
W_k & = & \Pi_{j=(k-1)l+1}^{kl}M_j,~ k \geq 1,~M_j \in \mathcal{M}, \label{e_W}
\end{eqnarray}
from where it follows that
\begin{eqnarray}
T_{lk} & = & \Pi_{j=1}^k ~ W_k,~ k \geq 1.   \label{e_T_W}
\end{eqnarray}
Observe that the set of time steps ``covered'' by $W_i$ and $W_j$, $i\neq j$,
are non-overlapping.  It is also important to note for subsequent analysis that, since the $M_k$'s are row stochastic matrices  and the product of any number of row stochastic matrices is row stochastic, all the $W_k$'s and $T_k$'s are also row stochastic matrices.


Lemma~\ref{l_t}  will establish  that as the number of iteration steps goes to infinity, the rows of the matrix $T_k$ tend to equalize. For proving Lemma~\ref{l_t}, we need the result in Lemma~\ref{l_w} stated below, which establishes that there exists a nonzero probability of choosing matrices in $\mathcal{M}$ such that the $W_k$'s as defined in  (\ref{e_W}) are scrambling.

\begin{lemma}
\label{l_w}
There exist constants $w>0$ and $d<1$ such that,
with probability equal to $w$, $\lambda(W_k)\leq  d$ for $k\geq 1$, independently for different $k$.
\end{lemma}

\begin{proof}
Each $W_k$ matrix is a product of
$l$ matrices from the set $\scriptm$.
The choice of the $M_k$'s that form $W_i$ and $W_j$ is independent
for $i\neq j$, since $W_i$ and $W_j$ ``cover'' non-overlapping intervals of
time. Thus, under the $i.i.d.$ assumption for selection of matrices
from $\scriptm$ \big(property (M4)\big), and property (M5), it follows that, with a non-zero probability (independently for $W_k$ and $W_{k'}$ for $k \neq k'$), matrix $W_k$ for each $k$ is scrambling. Let us denote by $w$ the probability that $W_k$ is scrambling.

Let us define $\scriptw$ as the set of all possible instances of $W_k$ that are scrambling. The set $\scriptw$ is finite because the set $\scriptm$ is finite,
and $\scriptw$ is also non-empty (this follows from the discussion of
(M5)). Let us define $d$ as the tight upper
bound on $\lambda(W)$, for $W\in\scriptw$, i.e.,
\begin{eqnarray}
d ~\equiv~\max_{W\in \scriptw}~ \lambda(W) \label{e_d}.
\end{eqnarray}
Recall that $\lambda(A)$
for any scrambling matrix $A$ is strictly less than 1.
Since $\scriptw$ is non-empty and finite, and contains only scrambling matrices,
it follows that
\begin{align}
d < 1.
\end{align}
 \end{proof}

\begin{lemma}
\label{l_t}
 There exist constants $\alpha$ and $\beta$ ($0<\alpha<1$, $0\leq\beta<1$) such that,
with probability greater than $ (1-\alpha^{k})$, $\delta(T_k)\leq \beta^k$ for $k\geq 8l/w$.
\end{lemma}
\begin{proof}
Let $k^* = \left\lfloor \frac{k}{l} \right\rfloor$
and $\Delta=k-lk^*$. Thus, $0\leq \Delta<l$.  From (\ref{e_M}) through (\ref{e_T_W}), observe that
\[ T_k =  T_{lk^*+\Delta} = T_{lk^*} ~ \Pi_{j=1}^{\Delta}~ M_{lk^*+j},
\]
where $T_{lk^*}$ is the product
of $k^*$ of $W_j$ matrices, where $1\leq j\leq k^*$.
As per Lemma~\ref{l_w}, for each $W_j$, the probability that
$\lambda(W_j)\leq d<1$ is equal to $w$.
Thus, the expected number of scrambling matrices among the $k^*$ matrices
is $w k^*$. Denote
by $S$ the actual number of scrambling $W_j$ matrices among the $k^*$ matrices.
Then the Chernoff lower tail bound tells us that,
for any $\phi>0$,
\begin{align}
&\Pr\{S < (1-\phi)E(S)\}   <   e^{-E(S) \phi^2/2}\\
\Rightarrow ~~~~~ &\Pr\{S < (1-\phi)(w k^*)\}  <  e^{-(w k^*) \phi^2/2}.
\end{align}
Let us choose $\phi=\frac{1}{2}$. Then,
\begin{align}
&\Pr\{S < (w k^*)/ 2\}   <   e^{-w k^* /8 } \\
\Rightarrow ~~~~ &\Pr\{S \geq w k^* / 2\}   >   1-e^{-w k^* /8 }.
\end{align}
Thus,
at least $\left\lfloor w k^*/2\right\rfloor$ of the $W$ matrices from the $k^*$ matrices
forming $T_{lk^*}$ are scrambling (each with $\lambda$ value $\leq d$,
by Lemma~\ref{l_w}) with probability greater than $ 1-e^{- w k^* /8 }$. 
Proposition~\ref{p_delta_product} then implies that
\[
\delta(T_k) = \delta(T_{lk^*+\Delta}) = 
 \delta\left(\left(\Pi_{i=1}^{k^*} W_i\right)
\left(\Pi_{i=1}^\Delta M_{lk^*+i}\right)\right)
	 \leq \left(\Pi_{i=1}^{k^*} \lambda(W_i)\right)\left(\Pi_{i=1}^\Delta \lambda(M_{lk^*+i})\right)
\]
Since at least $\left\lfloor w k^*/2\right\rfloor$ of the $W_i$'s have $\lambda(W_i)\leq d$
with probability greater than $ 1-e^{- w k^* /8 }$, and $\lambda(M_j)\leq 1,~\forall j$, it follows
that
\begin{eqnarray}
\delta(T_k) \leq d^{\left\lfloor w k^*/2\right\rfloor}
\label{e_l_delta}
\end{eqnarray}
with probability exceeding
\begin{eqnarray}
 1-e^{- w k^* /8 }.
\label{e_l_prob}
\end{eqnarray}
Let us define
$\alpha=e^{-\frac{ w}{16l}}$ and $\beta= d^\frac{ w}{8l}$.
Now, if $k\geq 8l/ w$, then if follows that $k \geq 2l$, and 
\begin{align}
k^* ~ =~&\left\lfloor \frac{k}{l}\right\rfloor  \geq   \frac{k}{2l}
\label{e_k_bound_2} \\
\Rightarrow ~~~~~
&\left\lfloor \frac{ w k^*}{2} \right\rfloor
   \geq  
\left\lfloor \frac{ w k}{4l} \right\rfloor \\
\Rightarrow ~~~~~
&\left\lfloor \frac{ w k^*}{2} \right\rfloor
   \geq   \frac{ w k}{8l} \\
\Rightarrow ~~~~~
&d^{\left\lfloor \frac{ w k^*}{2} \right\rfloor}   \leq   d^{\frac{ w k}{8l}} ~~~~~
	(\mbox{because~~} 0\leq d<1) \\
\Rightarrow ~~~~~ 
& d^{\left\lfloor \frac{ w k^*}{2} \right\rfloor}   \leq   \beta^k.
\label{e_l_d}
\end{align}
Similarly, if $k\geq 8l/w$, it follows that
\begin{align}
k^* ~ =~&\left\lfloor \frac{k}{l}\right\rfloor   \geq   \frac{k}{2l}
 \label{e_k_bound_1}\\
  \Rightarrow ~~~~~& e^{- w k^*/8}  
   \leq   e^{-\frac{ w k}{16l}} ~ = ~ \alpha^k \\
  \Rightarrow ~~~~~ 
& 1 - e^{- w k^*/8}   \geq  1 - \alpha^k.
\label{e_l_e}
\end{align}
By substituting (\ref{e_l_d}) and (\ref{e_l_e}) into
(\ref{e_l_delta}) and (\ref{e_l_prob}), respectively, the result  follows.
\end{proof}
Note that $\alpha$ and $\beta$ in Lemma~\ref{l_t} are independent of time. The threshold on $k$ for which Lemma~\ref{l_t} holds, namely
$k\geq 8l/ w$, can be improved by using better bounds
in (\ref{e_k_bound_2}) and (\ref{e_k_bound_1}). Knowing a smaller 
threshold
on $k$ for which Lemma~\ref{l_t} holds can be beneficial
in a practical implementation. In the above derivation for Lemma~\ref{l_t}, we chose a
somewhat loose threshold in order to maintain a simpler form
for the probability expression (namely, $1-\alpha^{k}$)
and also a simpler expression for the bound on $\delta(T_k)$
(namely, $\beta^{k}$).

\begin{lemma}
\label{l_as}
$\delta(T_k)$ converges almost surely to 0. 
\end{lemma}
\begin{proof}
For $k\geq 8l/w$, from Lemma~\ref{l_t}, we have that $\Pr\{\delta(T_k)> \beta^k\}\leq \alpha^{k}$, $0<\alpha<1$, $0\leq\beta<1$. Then, it is easy to see that $\sum_k\Pr\{\delta(T_k)> \beta^k\}\leq  8l/w +\sum_k  \alpha^{k} < \infty$. Then, by the first Borel-Cantelli lemma,
$
\Pr\{\mbox{the event that $\delta(T_k)> \beta^k$ occurs infinitely often}\} = 0
$.
Therefore,
$\delta(T_k)$ converges to 0 almost surely.
\end{proof}

\section{Convergence Analysis of Robustified Ratio Consensus Algorithm} \label{convergence_analysis}

The analysis below shows that the ratio  algorithm achieves asymptotic 
consensus correctly in the presence of the virtual nodes, even
if diagonals of the transition matrices ($M_k$'s) are not always strictly
positive. A key consequence is that the value of $z_k[i]$ is not necessarily greater from zero (at least not for all $k$), which creates
some difficulty when calculating the ratio $y_k[i]/z_k[i]$. As noted earlier, aside from these differences,
our algorithm is similar to that analyzed in
\cite{Benezit:10}. Our proof  has some similarities
to the proof in \cite{Benezit:10}, with the differences accounting
for our relaxed assumptions.

By defining  $z_k$ in an analogous way as we defined state $y_k$ in Section~\ref{InhMC}, the robustified version of the ratio consensus algorithm in \eqref{e_y_fixed}--\eqref{e_z_fixed} can be described in matrix form as 
\begin{eqnarray}
y_{k} & = &  y_{k-1} ~ M_k  ,~~~~~~~~~~ k \geq 1 \label{e_y_vary}, \\ 
z_{k} & = &  z_{k-1} ~ M_k   ,~~~~~~~~~~ k \geq 1 \label{e_z_vary},
\end{eqnarray}
where  $M_k  \in \mathcal{M},~k\geq 1$, $y_0[i] \geq 0,~\forall i$, $z_0[i]\geq 0,~\forall i$, and
  $\sum_j z_0[j]>0$, and $y_0[(i,j)]=z_0[(i,j)]=0,~\forall (i,j)\in\scripte$.
The same matrix $M_k$ is used at step $k$ of the  iterations in \eqref{e_y_vary} and \eqref{e_z_vary},
however, $M_k$ may vary over $k$.  
Recall that  $y_k$ and $z_k$  in (\ref{e_y_vary}) and \eqref{e_z_vary} have $n$ elements, but only the first $m$ elements correspond to computing nodes in the augmented network $\mathcal{G}^a$; the remaining entries in $y_k$ and $z_k$ correspond to virtual buffers. 

The goal of the algorithm is for each computing node to obtain a  consensus value defined as
\begin{align}
\pi^*=\frac{ \sum_j ~ y_0[j] }{\sum_j ~ z_0[j]} \label{ratio_cons}.
\end{align}
To achieve this goal, each node $i\in\scriptv$ calculates
\begin{eqnarray}
\pi_k[i] &=& \frac{y_k[i]}{z_k[i]},  \label{e_iter}
\end{eqnarray}
whenever the denominator is large enough, i.e., whenever
\begin{eqnarray}
 z_k[i]&\geq&\mu, \label{e_mu}
\end{eqnarray}
for some constant $\mu>0$ 
{\em to be defined later}.
We will show that, for each $i=1,2,\dots,m$, the sequence
$\pi_k[i]$ thus calculated asymptotically converges to 
the desired consensus value $\pi^*$. To show this, we first establish that \eqref{e_mu} occurs infinitely often, thus computing nodes can calculate  the ratio in \eqref{e_iter} infinitely often. Then, we will show that as $k$ goes to infinity, the sequence of ratio computations in  \eqref{e_iter} will converge to the value in \eqref{ratio_cons}.

The convergence when $\sum_j y_0[j]=0$ can be shown trivially.
So let us now consider the case when $\sum_j y_0[j]>0$, and
define new state variables $\newy_k$ and $\newz_k$  for
$k\geq 0$ as follows:
\begin{eqnarray}
\newy_k[i] & = & \frac{y_k[i]}{\sum_j y_0[j]}, ~ \forall i,\label{e_newy} \\
\newz_k[i] & = & \frac{z_k[i]}{\sum_j z_0[j]}, ~ \forall i. \label{e_newz}
\end{eqnarray}
Thus, $\newy_0$ and $\newz_0$ are defined by normalizing $y_k$ and $z_k$.
It follows that $\newy_0$ and $\newz_0$ are stochastic row vectors.
Also, since our transition matrices are row stochastic, it
follows that $\newy_k$ and $\newz_k$ are also stochastic vectors for all $k\geq 0$.

We assume that each node knows a lower bound on
$\sum_j z_0[j]$, denoted by $\mu_z$.
In typical scenarios, for all $i\in\scriptv$, $z_0[i]$ will be positive,
and, node $i\in\scriptv$ can use $z_0[i]$ as a non-zero lower bound on
$\sum_j z_0[j]$ (thus, in general, the lower bound used by different
nodes may not be identical). We also assume an upper bound, say $\mu_y$, on $\sum_j y_0[j]$.

Let us define 
\begin{eqnarray} \mu = \frac{\mu_z ~ c^l} {n}. \label{e_mu_def}
\end{eqnarray}
As time progresses, each node $i\in\scriptv$ will calculate
a new estimate of the consensus value whenever $z_k[i]\geq \mu$. The next lemma establishes that nodes will can carry out this calculation infinitely often.

\begin{lemma}
\label{l_scriptt}
Let  $\scriptt_i=\{\tau_i^1,\tau_i^2,\cdots\}$ denote  the sequence of
time instances when node $i$ updates its estimate of the consensus using \eqref{e_iter}, and obeying \eqref{e_mu},
where $\tau_i^j<\tau_i^{j+1}$, $j\geq 1$.
The sequence $\scriptt_i$ contains infinitely many elements with probability 1.
\end{lemma}
\begin{proof}
To prove the lemma, it will suffice to prove that for infinitely many values of
$k$, $z_k[i]>\mu$, with probability 1. Assumptions (M1)-(M5) imply that each
matrix $W_j,~j\geq 1$ (defined in (\ref{e_W}))
contains a strictly positive column corresponding to index $i\in\scriptv$
with a non-zero probability, say $\gamma_i>0$. Also,
the choice of $W_{k_1}$ and $W_{k_2}$
is independent of each other for $k_1\neq k_2$.
Therefore, the second Borel-Cantelli lemma implies that,
with probability 1,
for infinitely many values of $j$, $W_j$ will have the
$i$-th column strictly positive.
Since the non-zero elements of each matrix
in $\scriptm$ are all greater or equal to $c$, $c>0$ (by property M3),
and since $W_j$ is a product of $l$ matrices in $\scriptm$,
it follows that all the non-zero elements of each $W_j$
must be lower bounded by $c^l$.

Consider only those $j\geq 1$ for which $W_j$ contains positive $i$-th
column.
As noted above, there are infinitely many such $j$ values. 
Now,
\[
\newz_{jl} ~ = ~ \newz_{(j-1)l}~W_j.
\]
As noted above, $\newz_k$ is a stochastic vector.
Thus, for any $k\geq 0$,
\begin{eqnarray}
\sum_i \newz_{k}[i] & = & 1 
\end{eqnarray}
and, at least one of the elements of
$\newz_{(j-1)l}$ must be greater or equal to $1/n$.
Also, all the elements in columns of $W_j$ indexed by
$i\in\scriptv$ are lower bounded by $c^l$ (recall that we are
now only considering those $j$ for which the $i$-th column of $W_j$ is positive).
This implies that,
\begin{eqnarray}
\newz_{jl}[i] &  \geq &  c^l/n \\
\Rightarrow ~~~~~
z_{jl}[i] &  \geq & \left(\sum_j z_0[j]\right) ~ c^l/n \\ 
\Rightarrow ~~~~~ z_{jl}[i] &  \geq & \mu_z~c^l/n \\
\Rightarrow ~~~~~ z_{jl}[i] &  \geq & \mu, ~~ \forall i\in\scriptv ~~~~~~~~~~\mbox{(by (\ref{e_mu_def}))}
\label{e_event_a}
\end{eqnarray}
Since infinitely many $W_j$'s will contain a positive $i$-th column (with probability 1),
(\ref{e_event_a}) holds for infinitely many $j$ with probability 1.
Therefore, with probability 1,
the set $\scriptt_i=\{\tau_i^1,\tau_i^2,\cdots\}$ contains infinitely many elements,
for all $i\in\scriptv$.
\end{proof}

Finally, the next theorem shows that the ratio consensus algorithm will
converge to the consensus value defined in \eqref{ratio_cons}.

\begin{theorem}
\label{t_1}
Let $\pi_i[t]$ denote node $i$'s estimate of the consensus value calculated at time $\tau_i^t$.
For each node $i\in\scriptv$, with probability 1, the estimate
$\pi_i[t]$ converges to 
\[
\pi^* = \frac{\sum_jy_j[0]}{\sum_j z_j[0]}.
\]
\end{theorem} 
\begin{proof}
Note that the transition matrices $M_k,~k\geq 1$, are randomly drawn
from a certain distribution.
By an ``execution'' of the algorithm, we will mean a particular instance
of the $M_k$ sequence. Thus, the distribution on $M_k$'s results in
a distribution on the executions.
Lemma~\ref{l_as} implies that,
\[
\Pr \left\{\lim_{k\rightarrow\infty}\, \delta(T_k)=0\right\} = 1.
\]
Together, Lemmas~\ref{l_as} and \ref{l_scriptt} imply that, with
probability 1, for a chosen execution, (i) for any $\psi>0$,
there exists a finite $k_{\psi}$ such that, for all $k\geq k_{\psi}$, 
$\delta(T_k)<\psi$, and (ii) there exist infinitely many
values of $k\geq k_{\psi}$ such that $z_k[i]\geq\mu$ (i.e., $k\in\scriptt_i$ for
the chosen execution). 

Consider any $k\geq k_{\psi}$ such that $z_k[i]>\mu$.
Since $\delta(T_k)\leq \psi$,
the rows of matrix $T_k$ are ``within $\psi$'' of each other.
Observe that $\newy_k$ is obtained as the product
of stochastic row vector $\newy_0$ and
$T_k$. Thus, $\newy_k$ is in the convex hull of the
rows of $T_k$. Similarly $\newz_k$ is in the convex hull of the rows
of $T_k$. Therefore, the $j$-th elements of $\newy_k$ and $\newz_k$
are within $\psi$ of each other, for all $j$. Therefore,
\begin{eqnarray}
\left|~\newy_k[i] - \newz_k[i]~\right| & \leq & \psi \label{e_y_z}
\label{e_event_b} \\
\Rightarrow
\left|~\frac{\newy_k[i]}{\newz_k[i]}-1~\right| & \leq & \frac{\psi}{\newz_k[i]} \\
\Rightarrow ~~~~~
\left|~\frac{y_k[i]}{z_k[i]}-  
\frac{\sum_j y_0[j]}{\sum_j z_0[j]} ~\right| & \leq &
\frac{\psi ~ \sum_j y_0[j]}{z_k[i]} \mbox{~~~~(by
		(\ref{e_newy}) and (\ref{e_newz}))} \\
\Rightarrow ~~~~~ \left|~\frac{y_k[i]}{z_k[i]}-  
\frac{\sum_j y_0[j]}{\sum_j z_0[j]} ~\right|
& \leq &
\frac{\psi ~ \mu_y}{z_k[i]}
\mbox{~~~~~(because $\sum_j y_0[j]\leq \mu_y$)}
\label{e_bound_2} \\
\Rightarrow ~~~~~ \left|~\frac{y_k[i]}{z_k[i]}-  
\frac{\sum_j y_0[j]}{\sum_j z_0[j]} ~\right|
& \leq &
\frac{\psi ~ \mu_y}{\mu}.
\label{e_epsilon_new}
\end{eqnarray}
Now, given any $\epsilon>0$, let us choose
$\psi = \epsilon\mu/\mu_y$. Then
(\ref{e_epsilon_new}) implies that
\[
\left|~\frac{y_k[i]}{z_k[i]}-  
\frac{\sum_j y_0[j]}{\sum_j z_0[j]} ~\right|
 \leq
\epsilon
\]
whenever $k\geq k_{\psi}$ and $k\in\scriptt_i$.
Therefore, in the limit, 
$\frac{y_k[i]}{z_k[i]}$ for $k \in \mathcal{T}_i$
converges to $\frac{\sum_j y_0[j]}{\sum_j z_0[j]}$.
This result holds with probability 1, since conditions (i) and (ii)
stated above hold with probability 1.
\end{proof}

The result above can be strengthened by 
proving convergence of the algorithm even if each node $i\in\scriptv$
updates its estimate whenever $z_k[i]>0$ (not necessarily $\geq \mu$).
To prove the convergence in this case, the argument is similar to that
in Theorem~\ref{t_1}, with two modifications:
\begin{itemize}
\item Lemma~\ref{l_scriptt} needs to be strengthened by observing that there
exist infinitely many time instants at which $z_k[i]>\mu$ simultaneously for all $i\in\scriptv$.
This is true due to the existence of a matrix $T^*$ (as seen in the discussion of (M5))
that contains positive columns corresponding to all $i\in\scriptv$.
\item Using the above observation, and the argument in the proof
of Theorem~\ref{t_1}, it then follows that, with probability 1,
for any $\psi$, there exists a finite $k_{\psi}$ such that
$\delta(T_k)<\psi$ whenever $k\geq k_{\psi}$. As before,
defining $\psi = \epsilon\mu/\mu_y$, it can be shown that
for any $\epsilon$, {\em there exists}\, a $k_\epsilon\geq k_{\psi}$ such
that the following inequality holds for all $i\in\scriptv$ {\em simultaneously}.
\begin{eqnarray}
\left|~\frac{y_{k_\epsilon}[i]}{z_{k_\epsilon}[i]}-  
\frac{\sum_j y_0[j]}{\sum_j z_0[j]}~\right| 
 & \leq &
\epsilon ~~~~ \forall i\in\scriptv \\
\Rightarrow ~~~~
\frac{\sum_j y_0[j]}{\sum_j z_0[j]} - \epsilon
~~\leq~~ 
\frac{y_{k_\epsilon}[i]}{z_{k_\epsilon}[i]}
& \leq &
\frac{\sum_j y_0[j]}{\sum_j z_0[j]} + \epsilon
~~~~ \forall i\in\scriptv \label{e_eps}
\end{eqnarray}
\end{itemize}
Naturally, $z_{k_\epsilon}[i]\neq 0$, $\forall i\in\scriptv$.
It is now easy to argue that the above inequality will continue to hold
for all $k>k_\epsilon$ and each $i\in\scriptv$ whenever $z_i[k]>0$.
To see this, observe that, for $k>k_\epsilon$,
\[
y_k = y_{k_\epsilon}~ \Pi_{k_\epsilon+1}^k~ M_k
\] 
Define $P=\Pi_{k_\epsilon+1}^k~ M_k$ and 
$\psi = \epsilon\mu/\mu_y$.
Then, we have that, whenever $z_k[i]>0$ for $k>k_\epsilon$,
\begin{eqnarray}
\frac{y_k[i]}{z_k[i]} & = &
\frac{ \sum_{j=1}^n y_{k_\epsilon}[j]~P[j,i] }
 { \sum_{j=1}^n z_{k_\epsilon}[j]~P[j,i] } \\
& = & 
\frac{ \sum_{j=1, P[j,i]\neq 0}^n y_{k_\epsilon}[j]~P[j,i] }
 { \sum_{j=1, P[j,i]\neq 0}^n z_{k_\epsilon}[j]~P[j,i] }
~~~~~\mbox{(summation over non-zero $P[j,i]$ terms)} \nonumber \\
&\Rightarrow&
\min_{j,P[j,i]>0}~\frac{y_{k_\epsilon[j]}}{z_{k_\epsilon[j]}}
~~\leq~~ \frac{y_k[i]}{z_k[i]}
 \leq  
\max_{j,P[j,i]>0}~\frac{y_{k_\epsilon[j]}}{z_{k_\epsilon[j]}}
\label{e_ratio_ineq}
\\
& \Rightarrow &
\frac{\sum_j y_0[j]}{\sum_j z_0[j]} - \epsilon
~~\leq~~ \frac{y_k[i]}{z_k[i]}
 \leq  
\frac{\sum_j y_0[j]}{\sum_j z_0[j]} + \epsilon ~~~~ \mbox{from (\ref{e_eps})}
\\
& \Rightarrow &
\left|~\frac{y_{k}[i]}{z_{k}[i]}-  
\frac{\sum_j y_0[j]}{\sum_j z_0[j]}~\right| 
 \leq  \epsilon ~~~~\mbox{for all $i\in\scriptv$ and $k\geq k_\epsilon$}
\end{eqnarray}
This proves the convergence of the algorithm in the limit.
Recall that for this convergence it suffices if each node updates
its estimate of the consensus whenever its $z$ value is positive.
(\ref{e_ratio_ineq}) follows from the observation that
$\frac{\sum_j a[j]u[j]}{\sum_jb[j]u[j]}~=~\sum_j
\left[\frac{a[j]}{b[j]}
 ~ \frac{b[j]u[j]}{\sum_kb[k]u[k]}\right]$ is a weighted average
of $\frac{a[j]}{b[j]}$, and therefore, lower bounded by $\min_j \frac{a[j]}{b[j]}$ and upper bounded
by $\max_j \frac{a[j]}{b[j]}$.

\section{Concluding Remarks and Future Work} \label{concluding_remarks}

Although our analysis above is motivated by wireless environments
wherein transmissions may not succeed, the analysis is more general.
In particular, it applies to other situations in which properties
(M1)--(M5) are true. Indeed, property (M4) by itself is not
as important as its consequence that a given $W_k$ matrix
has non-zero columns indexed by $i\in\scriptv$.

A particular application of the above analysis is in the case
when messages may be delayed.
As discussed previously, mass is transfered by any node
to its neighbors by means of messages. Since these messages
may be delayed, a message sent on link $(i,j)$ in slot $k$
may be received by node $j$ in a later slot.
Let us denote by $V_k[i]$ the set of
messages received by node $i$ at step $k$.
It is possible for $V_k[i]$ to contain
multiple messages from the same node.
Note that $V_k[i]$ may contain a message
sent by node $i$ to itself as well. Let us define
the iteration for $y_k$ as follows:

\begin{eqnarray}
y_k[i] & = & \sum_{v\in V_k[i]} v.
\label{e_async}
\end{eqnarray}
The iteration for $z_k$ can be defined analogously.
Our robust consensus algorithm essentially implements
the above iteration, allowing for delays in delivery of
mass on any link $(i,j)$ (caused by link failures). However,
in effect, the robust algorithm also ensures FIFO (first-in-first-out)
delivery, as follows. In slot $k$, if node $i$ receives mass sent by
node $j\in \scripti_i$ in slot $s,~s\leq k$, then mass sent by node $j$
in slots strictly smaller than $s$ is either received previously, or will be received
in slot $k$.

The virtual buffer mechanism essentially models asynchronous communication,
wherein the messages between any pair of nodes in the network may require
arbitrary delay, governed by some distribution.
It  is not difficult to see that the iterative algorithm (\ref{e_async})
should be able to achieve consensus correctly even under other
distributions on message delays, with possible correlation between
the delays. In fact, it is also possible to
tolerate non-FIFO (or out-of-order) message delivery provided
that the delay distribution satisfies some reasonable constraints.
Delay of up to $B$ slots on a certain link $(i,j)\in\scripte$ can be modeled
using a single chain of $B$ virtual nodes, with links from node $i$
to every virtual nodes, and link from the last of the $B$
nodes to node $j$---in this setting, depending on the delay incurred
by a packet, appropriate link from node $i$ to one of the virtual node  on
the delay chain (or to   $j$, if delay is 0) is used.

Note that while we made certain assumptions regarding link failures,
the  analysis relies primarily on two implications of these
assumptions, namely (i) the rows of the transition matrix $T_k$ become
close to identical as $k$ increases, and (ii) $z_k[i]$ is bounded
away from 0 for each $i$ infinitely often.
When these implications are true, similar convergence results may hold
in other environments.

\bibliographystyle{IEEEtran}

\bibliography{distributed,distributed2}

\begin{thebibliography}{1}
\providecommand{\url}[1]{#1}
\csname url@samestyle\endcsname
\providecommand{\newblock}{\relax}
\providecommand{\bibinfo}[2]{#2}
\providecommand{\BIBentrySTDinterwordspacing}{\spaceskip=0pt\relax}
\providecommand{\BIBentryALTinterwordstretchfactor}{4}
\providecommand{\BIBentryALTinterwordspacing}{\spaceskip=\fontdimen2\font plus
\BIBentryALTinterwordstretchfactor\fontdimen3\font minus
  \fontdimen4\font\relax}
\providecommand{\BIBforeignlanguage}[2]{{%
\expandafter\ifx\csname l@#1\endcsname\relax
\typeout{** WARNING: IEEEtran.bst: No hyphenation pattern has been}%
\typeout{** loaded for the language `#1'. Using the pattern for}%
\typeout{** the default language instead.}%
\else
\language=\csname l@#1\endcsname
\fi
#2}}
\providecommand{\BIBdecl}{\relax}
\BIBdecl

\bibitem{Se:06}
E.~Seneta, \emph{Non-negative Matrices and Markov Chains}, revised
  printing~ed.\hskip 1em plus 0.5em minus 0.4em\relax New York, NY: Springer,
  2006.

\bibitem{Benezit:10}
F.~Benezit, V.~Blondel, P.~Thiran, J.~Tsitsiklis, and M.~Vetterli, ``Weighted
  gossip: Distributed averaging using non-doubly stochastic matrices,'' in
  \emph{Proc. of IEEE International Symposium on Information Theory}, June
  2010, pp. 1753 --1757.

\bibitem{KeDo:03}
D.~Kempe, A.~Dobra, and J.~Gehrke, ``Gossip-based computation of aggregate
  information,'' in \emph{Proc. IEEE Symposium on Foundations of Computer
  Science}, Oct. 2003, pp. 482 -- 491.

\bibitem{DoHa:10}
A.~D. Dom\'inguez-Garc\'ia and C.~N. Hadjicostis, ``Coordination and control of
  distributed energy resources for provision of ancillary services,'' in
  \emph{Proc. IEEE SmartGridComm}, 2010, pp. 537 -- 542.

\bibitem{Cao.Morse.ea2008}
M.~Cao, A.~S. Morse, and B.~D.~O. Anderson, ``Reaching a consensus in a
  dynamically changing environment: Convergence rates, measurement delays, and
  asynchronous events,'' \emph{SIAM Journal on Control and Optimization},
  vol.~47, no.~2, pp. 601--623, 2008.

\bibitem{Nedic:2010}
A.~Nedi\'{c} and A.~Ozdaglar, ``Convergence rate for consensus with delays,''
  \emph{Journal of Global Optimization}, vol.~47, pp. 437--456, July 2010.

\bibitem{TsRa:11}
K.~Tsianos and M.~Rabbat, ``Distributed consensus and optimization under
  communication delays,'' in \emph{Proceedings of Annual Allerton Conference on
  Communication, Control, and Computing}, September 2011.

\bibitem{Ha:58}
J.~Hajnal, ``Weak ergodicity in non-homogeneous \text{M}arkov chains,''
  \emph{Proceedings of the Cambridge Philosophical Society}, vol.~54, pp. pp.
  233--246, 1958.

\bibitem{Wo:63}
J.~Wolfowitz, ``\BIBforeignlanguage{English}{Products of indecomposable,
  aperiodic, stochastic matrices},''
  \emph{\BIBforeignlanguage{English}{Proceedings of the American Mathematical
  Society}}, vol.~14, no.~5, pp. pp. 733--737, 1963.

\end{thebibliography}

\end{document}